\documentclass[copyright]{eptcs}

\usepackage[utf8x]{inputenc}
\usepackage{latexsym,amssymb,amsmath,amsthm}

\newcounter{theorem}
\newtheorem{proposition}[theorem]{Proposition}
\newtheorem{result}[theorem]{Result}
\theoremstyle{definition}
\newtheorem{example}[theorem]{Example}

\title{Descriptional Complexity of the Languages $KaL$:
Automata, Monoids and Varieties\thanks{Both authors were supported by the Ministry of Education of the
Czech Republic under the project MSM 0021622409 and by
the Grant 201/09/1313 of the Grant Agency of the Czech Republic.}}
     
\author
{Ond\v rej Kl\'{i}ma \qquad\qquad \qquad Libor Pol\'ak
\institute{Department of Mathematics and Statistics\\ Masaryk University\\
Brno, Czech Republic}
\email{\quad klima@math.muni.cz\quad\qquad polak@math.muni.cz}
}
\def\titlerunning{Descriptional Complexity of the Languages $KaL$}
\def\authorrunning{O. Kl\'{i}ma  \&  L. Pol\'ak}
\begin{document}
\maketitle

\begin{abstract}
The first step when forming the polynomial hierarchies of languages
is to consider
languages of the form $KaL$ where $K$ and $L$ are over a finite alphabet $A$
and from a given variety $\mathcal V$
of languages,  $a\in A$ being a letter. All such $KaL$'s generate the  variety
of languages $\mathsf{BPol}_1(\mathcal V)$.

We estimate the numerical parameters of the language $KaL$
in terms of their values for $K$ and $L$. These parameters include 
the state complexity
of the minimal complete DFA and  the size of the syntactic monoids.
We also estimate 
the cardinality of the image of $A^*$
in the Sch\"utzenberger product of the syntactic monoids
of $K$ and $L$. In these three cases we obtain the optimal bounds.

Finally, we also consider estimates for 
the cardinalities of free monoids
in the variety of monoids corresponding to $\mathsf{BPol}_1(\mathcal V)$
in terms of sizes of the free monoids in the variety
of monoids corresponding to~$\mathcal V$.
\end{abstract}

\section{Introduction}

The polynomial operator assigns to each variety of languages $\mathcal V$
the class of all Boolean combinations
of the languages of the form
$$L_0a_1 L_1a_2\dots a_\ell L_\ell\ , \qquad\qquad (\ast)$$
where $A$ is a finite alphabet,
$a_1,\dots,a_\ell\in A,\, L_0,\dots,L_\ell\in \mathcal V(A)$ (i.e. they are
over $A$).
Such operators on classes of languages lead to several concatenation hierarchies.
Well known cases are the
Straubing-Th\'erien and the group hierarchies. Concatenation hierarchies
has been intensively studied by many authors -- see Section 8 of the
Pin's Chapter \cite{pi-kap}.
In the restricted case we fix a natural number $k$ and we allow only
$\ell\leq k$ in $(*)$ -- see \cite{kp-cai} and papers quoted there.
The resulting variety of languages is denoted by
$\mathsf{BPol}_k(\mathcal V)$. Using the Eilenberg correspondence, 
$\mathsf{BPol}_k$ operates also on pseudovarieties of monoids.
 We consider in this paper only the case $k=1$.

State complexity problems are a fundamental part of automata theory. 
Recent papers of a survey nature with numerous references are \cite{b}
by Brzozowski and \cite{y} by Yu.
First we estimate the state complexity of DFA automata for the language
$KaL$ in terms of the state complexities of $K$ and $L$. This is the content
of Section~2.

Secondly, for languages $K$ and $L$, we also estimate 
the cardinality of the image of $A^*$ under the natural homomorphism $\mu_a$
into the Sch\"utzenberger product of the syntactic monoids $M$ and $N$
of the languages
$K$ and $L$. This monoid $\mu_a(A^*)$ recognizes the language $KaL$, too.
The syntactic monoid of $KaL$
is a homomorphic image of the monoid $\mu_a(A^*)$.
The third question concerns its cardinality.

In all three problems we get estimates which can be reached by concrete
examples (for the first one in Section 2 and for the two remaining ones 
in Section~3).
In general: the size of the Sch\"utzenberger product equals at least
to the size of the monoid $\mu_a(A^*)$ which is at least the size
of the syntactic monoid of $KaL$.
In Section~3 we further consider natural examples showing that those
three numbers could differ drastically. The first example
is the language $B^*aC^*,\,B,C\subseteq A$.
The next proposition roughly estimates $\mu_a(A^*)$ for $\mathcal J$-trivial
monoids using their structure. 

In the last section we consider
a variety of languages $\mathcal V$ such that the corresponding
pseudovariety of monoids consists of all finite members of a locally finite
variety of monoids $\mathbf V$. Then the free monoid $F_{\mathbf V}(A)$
in $\mathbf V$
over a finite set $A$ is the smallest one recognizing
all languages in $\mathcal V(A)$.
We  embed the free monoid in the variety of monoids corresponding to
the class
$\mathsf{BPol}_1(\mathcal V)$ over $A$ 
into the product of
$|A|$ copies of the Sch\"utzenberger product of $F_{\mathbf V}(A)\Diamond
F_{\mathbf V}(A)$
which leads to a rough
estimate for the cardinality of this free monoid.

\section{Recognizing by Automata}

Let $A$ be a finite alphabet and let $L\subseteq A^*$ be a regular language.
The following construction of the minimal complete DFA is due to Brzozowski.
We put:
$\mathsf D=\{\, u^{-1} L \mid u\in A^*\, \}$ --
the set of all {\it left derivatives} of $L$
(here $u^{-1} L =\{\,v\in A^*\mid uv\in L\,\}$).
One assigns to $L$ its ``canonical'' {\it minimal automaton}
$\mathcal D = (\mathsf D,A,\cdot,L,F)$
using left derivatives; namely:
\begin{itemize}
\item $\mathsf D$ is the (finite) set of states,
\item $a\in A$ acts on $u^{-1} L$ by
$(u^{-1} L) \cdot a = a^{-1} (u^{-1} L)$,
\item $L$ is the initial state and $Q\in \mathsf D$ is a final state
(i.e., element of $F$ ) if and only if $1\in Q$.
\end{itemize}

\begin{proposition}
Let $K$ and $L$ be languages over a finite alphabet $A$ whose minimal
complete DFA have $k$ resp. $\ell$ states and let $a\in A$.
Then the minimal complete
DFA for the language $KaL$ has at most $k2^\ell$ states.
\end{proposition}

\begin{proof}\label{prop-aut}
Notice that an arbitrary left derivative of $KaL$ is of the form
$$(u^{-1}K)aL\cup u_1^{-1}L\cup\cdots\cup u_m^{-1}L,\ m
\text{ a non-negative integer},\, u_1,\dots,u_m\in A^*\,.$$
We have $k$ possible values for $u^{-1}K$ and 
$u_1^{-1}L\cup\cdots\cup u_m^{-1}L$, $m\leq \ell$,
has at most $2^\ell$ values.
The statement follows.
\end{proof}

The example in the next proposition is a slight modification of the construction 
in Theorem 2.1
in \cite{yzs}. It was suggested to the authors by J. Brzozowski. It shows
that the bound from Proposition~1 is tight. 

\begin{proposition}
For arbitrary natural numbers $k,\ell\geq 2$ there exist languages $K$ resp. $L$
whose
minimal complete DFA have $k$ resp. $\ell$ states such that each complete
DFA recognizing the language $KaL$ has at least $k2^\ell$ states.
\end{proposition}

\begin{proof}
 Let $A=\{a,b,c\}$ and
let $\mathcal A= (\,\{p_0,\dots,p_{k-1}\}, A, \cdot, p_0, \{p_{k-1}\}\,)$ 
where
$$p_i\cdot a =p_0,\ p_i\cdot b= p_{i'} \text{ where } i'\equiv i+1\ (\text{mod } k),\ 
p_i\cdot c =p_i \text{ for all } i=0,\dots,k-1\,.$$
Note that $\mathcal A$ accepts the language
$$K=\mathsf L(\mathcal A)
=\{\,uv\mid u\in(A^*a)^*,\ v\in \{b,c\}^*,\text{ and }
|v|_b \equiv k-1\ (\text{mod }k)\,\}\ .$$
Similarly, we define 
$\mathcal B=\{\,\{q_0,\dots,q_{\ell-1}\}, A, \cdot, q_0,\{q_{\ell-1}\}\,\}$
where
$$q_j\cdot a=q_{j'} \text{ where } j'\equiv j+1\ (\text{mod }\ell\,),\ 
q_j\cdot b=q_j,\ q_j\cdot c=q_1 \text{ for all } j=0,\dots,\ell-1\,.$$
Clearly, for $L=\mathsf L(\mathcal B)$, we have
$$L\cap \{a,b\}^*=
\{\,u\in \{a,b\}^*\mid |u|_a\equiv \ell-1\ (\text{mod } \ell\,)\,\}\ .
$$
Both automata $\mathcal A$ and $\mathcal B$ are minimal.
\smallskip

We define, for all $u\in \{a,b\}^*$, the set
$$S(u)=\{\,i\in \{0,\dots,\ell-1\} \mid u=vaw \text{ such that } v\in K 
\text{ and } 
i\equiv |w|_a\ (\text{mod } \ell)\,\}$$
and the numbers
$$T(u)= \text{ the greatest } m \text{ such that } b^m \text{ is
a suffix of } u $$
and
$$t(u)\in\{0,\dots,k-1\},\ t(u)\equiv T(u)\ (\text{mod } k\,)\,.$$
Let $u,v\in \{a,b\}^*$ be such that $S(u)\not=S(v)$. 
Let $s\in S(u)\setminus
S(v)$ (the case $s\in S(v)\setminus S(u)$ can be treated similarly).
Then $ua^{\ell-1-s}\in KaL$ but $va^{\ell-1-s}\not\in KaL$. Then in each
complete DFA recognizing $KaL$ with the initial state $r_0$ we have
that
$$r_0\cdot u\not= r_0\cdot v\, .$$

Now let $u,v\in \{a,b\}^*$ be such that $S(u)=S(v)$ and
$t = t(u)>t(v)$. Then, for $w=cb^{k-1-t}a^{\ell}$,
we have
$uw\in KaL$ and $vw\not\in KaL$.
Again, in each
complete DFA recognizing $KaL$ with the initial state $r_0$ we have
that $$r_0\cdot u\not= r_0\cdot v\, .$$
For an arbitrary subset $S=\{s_1,\dots,s_m\}$ of $\{0,\dots,\ell-1\}$,
where $s_1>\dots>s_m$, and $t\in\{0,\dots,k-1\}$ there exists a word
$$u=b^{k-1}a^{s_1-s_2}b^{k-1}a^{s_2-s_3}b^{k-1}\dots b^{k-1}
a^{s_m+1}b^t$$
such that $S(u)=S$ and $t(u)=t$.

Therefore each
complete DFA recognizing $KaL$ has at least $k2^\ell$ states.
\end{proof}

\section{Recognizing by Monoids} 

Let $K$ and $L$ be languages over a finite alphabet $A$ and let
$M$ and $N$ be their syntactic monoids.
In this section we will compare 
\begin{itemize}
\item 
the size of the Sch\"utzenberger product $M\Diamond N$
of monoids $M$ and $N$,
\item
the cardinality of the image of $A^*$ in the homomorphism $\mu_a$
from   $A^*$
into  $M\Diamond N$ recognizing the language $KaL$,
\item
the size of the syntactic monoid of the language $KaL$.
\end{itemize}

Let $M$ and $N$ be finite monoids. Their {\it Sch\"utzenberger product}
$M\Diamond N$ is the set of all 2$\times$2 matrices 
$P$ where
$P_{2,1}=\emptyset,\ P_{1,1}\in M,\ P_{2,2}\in N$ and
$P_{1,2}\subseteq M\times N$
equipped with the multiplication
$$(PQ)_{1,1}=P_{1,1}Q_{1,1},\ (PQ)_{2,2}=P_{2,2}Q_{2,2}\quad
\text{ and } 
$$$$
(PQ)_{1,2}= \{\,(P_{1,1}x,y) \mid (x,y)\in Q_{1,2}\,\}\ \cup
\{\,(z,tQ_{2,2}) \mid (z,t)\in P_{1,2}\,\}\ .$$
It is well known that this operation is associative.
This product was introduced by Sch\"utzenberger and by Straubing for
an arbitrary finite family of monoids. Basic results are also due to
Reutenauer and Pin -- see \cite{pi-kniha} Theorems 1.4 and 1.5 in Chapter 5.
Clearly, if $|M|=m$ and $|N|=n$, then $|M\Diamond N|=mn2^{mn}$.
\medskip

Recall that the {\it syntactic congruence} of the language $R\subseteq A^*$
is a relation ${\sim}_R$ on $A^*$ defined by:
$$u \sim_R v\quad \text{ if and only if }\quad (\ \forall\ p,q\in A^*\ )\ 
(\ puq\in R\, \Leftrightarrow\,  pvq\in R\ )\,.$$ 
The {\it syntactic monoid} of $R$ is the quotient monoid
$A^*/{\sim_R}$. It is the smallest monoid recognizing the language $R$.

Let $A$ be a finite alphabet and let $\varphi : A^* \rightarrow M,\ 
\psi : A^* \rightarrow N$ be homomorphisms. Let $S\subseteq M,\ 
T\subseteq N$ and let $K=\varphi^{-1}(S),\ L=\psi^{-1}(T)$, i.e.
the language $K$ is {\it recognized} by $M$ using $\varphi$ and
$S$, and similarly for the language $L$. One can take the mappings
$\varphi$ and $\psi$ surjective.

For $a\in A$, we define a mapping $\mu_a : A^* \rightarrow M\Diamond N$ by
$$
(\mu_a(u))_{1,1}=\varphi(u),\ (\mu_a(u))_{2,2}=\psi(u)\quad
\text{ and }$$ 
$$
(\mu_a(u))_{1,2} =\{\, (\varphi(u'),\psi(u'')) \mid
u=u'au'',\ u',u''\in A^*\,\}\,.$$
It is easy to see that it is a homomorphism and that the language
$KaL$ is recognized by $M\Diamond N$ using
$\mu_a$ and $\{\,P\in M\Diamond N \mid  P_{1,2}\cap S\times T \not=
\emptyset\,\}.$
\smallskip

Of course, the language $KaL$ is also recognized by
$\mu_a(A^*)$ which can be much smaller than the whole  $M\Diamond N$.
Moreover the syntactic monoid of the language $KaL$ is a homomorphic
image of the monoid $\mu_a(A^*)$. Its size can be much smaller than 
the cardinality of the monoid $\mu_a(A^*)$.
\smallskip

First we present, for arbitrary $m$ and $n$, an example where 
the mapping $\mu_a$ is onto. Thus the bound $mn2^{mn}$ for $\mu_a(A^*)$
is sharp.

\begin{proposition}\label{prop-image-sharp} 
For arbitrary $m$ and $n$, there exist languages
 $K$ and $L$ with syntactic monoids
$M$ and $N$ and homomorphisms
$\varphi : u\mapsto u{\sim}_K,\ \psi : u \mapsto u{\sim}_L$, $ u\in A^*$, 
such that the mapping $\mu_a : A^* \rightarrow M\Diamond N$
is surjective.
\end{proposition}

\begin{proof} Let $A=\{a,b,c\}$,
let $m$ and $n$ be natural numbers and let
$$K=\{\,u\in A^*\mid |u|_b\equiv 0\!\!\! \pmod m\,\} \text{ and }
L=\{\,u\in A^*\mid |u|_c \equiv 0\!\!\! \pmod n\,\}\,.$$
The syntactic monoids of $K$ and $L$ are the additive groups
$\mathbb Z_m$ and  $\mathbb Z_n$ and the syntactic homomorphisms are
given by
$$\varphi(u)=[|u|_b]_m\in \mathbb Z_m \quad \text{ and }\quad
\psi(u)=[|u|_c]_n\in \mathbb Z_n\,.$$
Let $k\in\{0,\dots,m-1\}\,, \ell\in \{0,\dots,n-1\}$ and
$O\subseteq \{0,\dots,m-1\}\times \{0,\dots,n-1\}$ be arbitrary.
We will find $u\in A^*$ such that
$$|u|_b\equiv k\ (\text{mod } m),\ 
|u|_c \equiv \ell\ (\text{mod } n),\text{ and } 
$$
$$
\{\,(p,q)\in\{0,\dots,m-1\}\times \{0,\dots,n-1\} \mid
$$$$
|u'|_b\equiv p\ (\text{mod } m),\ 
|u''|_c \equiv q\ (\text{mod } n),\,
u=u'au'',\,u',u''\in A^*\,\}
=O\,.$$
Let
$$O=\{\,(0,j_{0,1}),\dots,(0,j_{0,p_0}),
(1,j_{1,1}),\dots,(1,j_{1,p_1}),
$$
$$\dots$$
$$
(m-1,j_{m-1,1}),\dots,(m-1,j_{m-1,p_{m-1}})\,  \}\,,$$
where
$$n-1\geq j_{0,1}>\dots>j_{0,p_0}\geq 0,\dots,
n-1\geq j_{m-1,1}>\dots>j_{m-1,p_{m-1}}\geq 0\,.$$
We put
$$u_0=c^{n-j_{0,1}} a c^{j_{0,1}-j_{0,2}} a \dots a 
c^{j_{0,p_0-1}-j_{0,p_0}} a c^{j_{0,p_0}}\,,$$
$$u_1=c^{n-j_{1,1}} a c^{j_{1,1}-j_{1,2}} a \dots a 
c^{j_{1,p_1-1}-j_{1,p_1}} a c^{j_{1,p_1}},$$
$$\dots$$
$$u_{m-1}=c^{n-j_{m-1,1}} a c^{j_{m-1,1}-j_{m-1,2}} a \dots a 
c^{j_{m-1,p_{m-1}-1}-j_{m-1,p_{m-1}}} a c^{j_{m-1,p_{m-1}}}$$
where $u_i=1$ if $p_i=0$, for $i=0,\dots,m-1$.
Finally, putting
$$u=c^\ell u_0bu_1b\dots u_{m-1}b\cdot b^k\,,$$
we see that this word has all desired properties.
\end{proof}

We used GAP to calculate the sizes of syntactic monoids from the last
proof for $m\in \{2,3,4\}$ and $n=2$. The numbers were
61, 379 and 2041. They are of the form $mn(2^{mn}-1)+1$. This led us
to the following two results.

\begin{proposition}\label{prop-groups+}
Let $K$ and $L$ be languages over a finite alphabet $A$ 
with syntactic monoids $M$ and $N$,
let $|M|=m,\,|N|=n$ and let $a\in A$.
Then the size of the syntactic monoid of $KaL$
is at most $mn(2^{mn}-1)+1$.
\end{proposition}
 
\begin{proof}
(i) Suppose first that both $M$ and $N$ are groups.
Let $u\in A^*$ be such that $(\mu_a(u))_{1,2} = M\times N$.
Then also, for each $p,q\in A^*$, it is the case that  
$(\mu_a(puq))_{1,2} = M\times N$.
Therefore, each pair $(u,v)\in A^*\times A^*$ with
$(\mu_a(u))_{1,2} = (\mu_a(v))_{1,2} = M\times N$
is in the syntactic congruence of the language $KaL$.
\smallskip

(ii) Suppose that the monoid $M$ is not a group (the case $N$ not being
a group could be treated in a similar way).
Let $s\in M$ be without an inverse element.  Then there is no $t\in M$ 
with $st=1$. Indeed, such $t$ would imply that $u \mapsto us,\ u\in M$, 
is one-to-one
and due to the finiteness of $M$ we have that $\{\,us\mid u\in M\,\}=M$. 
Thus there would
be $u\in M$ such that $us=1$ and $t=us\cdot t=u\cdot st=u$ -- a~contradiction.

Let $u\in A^*$ and let $(s,1)\in (\mu_a(u))_{1,2}$. Thus there exist
$u',u''\in A^*$ such that $u=u'au''$ and $\varphi(u')=s,\,\psi(u'')=1$.
Consequently
$(\mu_a(u))_{1,1}=\varphi(u)=s \varphi(a) \varphi(u'')\not=1$.
There are $mn2^{mn-1}$ matrices in $M\Diamond N$ not having the element
$(s,1)$ in the set at position $(1,2)$, and 
$(m-1)n2^{mn-1}$ matrices in $M\Diamond N$ having the element
$(s,1)$ in the set at position $(1,2)$ and not having $1$ at position $(1,1)$.

Consequently, the size of the syntactic monoid of $KaL$ is less
or equal the cardinality of $\mu_a(A^*)$ which is at most
$mn2^{mn-1}+ (m-1)n2^{mn-1}$. The gap between $mn2^{mn}$ and the last number
is at least the needed value $mn-1$.
\end{proof}

Next we show that the estimate from Proposition~\ref{prop-groups+} is exact.

\begin{proposition}\label{prop-synt-monoids-sharp}
For arbitrary $m$ and $n$, there exist languages $K$ and $L$ with
syntactic monoids $M$ and $N$, $|M|=m,\,|N|=n$, such that
the size of the syntactic monoid of $KaL$
is exactly $mn(2^{mn}-1)+1$.
\end{proposition}
 
\begin{proof}
We again consider the languages $K$ and $L$ from the proof of 
Proposition~\ref{prop-image-sharp}.

(i) Let $u,v\in A^*,\ ([k]_m,[\ell]_n)\in (\mu_a(u))_{1,2} \setminus
(\mu_a(v))_{1,2},\ k\in\{0,\dots,m-1\},\,\ell\in\{0,\dots,n-1\}$.

Let $$p=b^{m-k},\, q=c^{n-\ell}\,.$$
Then $puq\in KaL$ and $pvq\not\in KaL$.

\medskip
(ii) Let $u,v\in A^*,\ ([k]_m,[\ell]_n) \not\in (\mu_a(u))_{1,2} =
(\mu_a(v))_{1,2},\ k\in\{0,\dots,m~-~1\},\,\\
\ell\in\{0,\dots,n-1\}$.
Let $(\mu_a(u))_{1,1}\not=(\mu_a(v))_{1,1}$.
(The case  $(\mu_a(u))_{2,2}\not=(\mu_a(v))_{2,2}$ could be treated
analogously).

Let $p=b^{m-k},\, \text{ let } \alpha \text{ be a natural number such that }
\beta = \alpha m-|b^{m-k}u|_b\geq 0,\text{ and let } 
q=c^{n-\ell} b^\beta a$.\\
Then $puq\in KaL$ and $pvq\not\in KaL$.
\end{proof}

The following example shows that the cardinalities of $M\Diamond N$,
the cardinality of $\mu_a(A^*)$ and the size of the syntactic monoid 
can be three quite different numbers.
 
\begin{example}
Let $a\in A$, let $B,C\subsetneqq A$ and consider the language $B^*aC^*$.
Syntactic monoids of both $B^*$ and $C^*$ are isomorphic to the
two element monoid $2=\{0,1\}$ having a neutral element 1 and
a zero element 0. Moreover, for $a\in A$,
$\varphi(a)= 1$ if and only if $a\in B$, and
$\psi(a)= 1$ if and only if $a\in C$.
Finally $S=T=\{1\}.$

Clearly, the cardinality of $2\Diamond 2$ is 
$2\cdot 2 \cdot 2^{2\cdot 2}=64$.

Let $A=\{a,b,c,d\},\,B=\{a,b\}$ and $C=\{a,c\}$. 
One can calculate that $|\mu_a(A^*)|=22$.
Finally, it is well known and easy to see that
the syntactic monoid of $B^*aC^*$ is isomorphic to the 8-element
monoid of Boolean uppertriangular matrices of order 2.
\end{example}

\medskip
We will try to estimate the number $|\mu_a(A^*)|$ using the structures
of monoids $M$ and $N$. The first little step concerns very
special monoids and certain chains of their elements. 

Green's relations are a basic tool in semigroup theory:
define on an arbitrary monoid $O$ the quasiorders $\leq_\mathcal R$,
$\leq_\mathcal L$ and $\leq_\mathcal J$ as follows:
$$p\leq_\mathcal R q \text{ iff } p=qr \text{ for some } r,
\quad
p\leq_\mathcal L q \text{ iff } p=sq \text{ for some } s\,,
$$$$
\quad \text{ and }
p\leq_\mathcal J q \text{ iff } p=sqr \text{ for some } r,s
\,.$$
A monoid $O$ is $\mathcal J${\it -trivial} if 
$p\leq_\mathcal J q \leq_\mathcal J p$ implies that $p=q$.
For each $u\in A^*$, we define $\mathsf{c}(u)$ (the {\it content} of $u$)
as the set of all letters
of $u$.

\begin{proposition}\label{prop-chains}
Let $M$ and $N$ be finite $\mathcal J$-trivial monoids having cardinalities
$m$ and $n$.
Let the number of elements in
a longest strict $\leq_\mathcal R$-chain in $M$ is $\rho$ and the
number of elements in
a longest strict $\leq_\mathcal L$-chain in $N$ is $\lambda$. Let $a\in A$,
$\varphi : A^* \rightarrow M,\ 
\psi : A^* \rightarrow N$ be homomorphisms.
Then the number of elements of each set of $(\mu_a(u))_{1,2},\ u\in A^*$,
is less or equal to $\rho+\lambda-1$ (which is $\leq m+n=1$).
In particular,
$$|\mu_a(A^*)|\leq mn(\binom{mn}{0}+\binom{mn}{1}+\dots+
\binom{mn}{\rho+\lambda-1})\,.$$
\end{proposition}

\begin{proof}
Let $u=u_0au_1au_2\dots au_k$ where $a\not\in \mathsf{c}(u_0),
\mathsf{c}(u_1),\dots,\mathsf{c}(u_k)$. Then
$$(\mu_a(u))_{1,2}=
\{\,(\varphi(u_0),\psi(u_1au_2\dots au_k)),\,
(\varphi(u_0au_1),\psi(u_2au_3\dots au_k)),\dots$$
$$
\dots,
(\varphi(u_0au_1\dots au_{k-1}),\psi(u_k))\,\}$$
and the statement follows.
\end{proof}

The following example shows that the bound for $|(\mu_a(u))_{1,2}|$ from
Proposition~\ref{prop-chains} is sharp.

\begin{example}
For $B\subseteq A$ we write 
$\overline B=\{\,u\in A^*\mid \mathsf{c}(u)=B\,\}.$
Notice first that, for $B\subsetneqq A$,
the syntactic monoid of $\overline B$ is isomorphic to
the monoid $(2^B,\cup)$ with a zero 0 adjoined.  
The syntactic homomorphism $\varphi$ maps $u\in B^*$ onto $\mathsf{c}(u)$
and $\varphi(u)=0$, otherwise, and we have $S=\{B\}$.
All the relations $\leq_\mathcal R,\, \leq_\mathcal L,\,
 \leq_\mathcal J$ coincide with the reverse inclusion
$\supseteq$.
Similarly for $C\subsetneqq A$.

Consider first the language $\overline B a \overline C$ for
$A=\{a,b,c\},\,B=\{a,b\},\, C=\{a,c\}$.
Then $\rho=\lambda =4$ and $\rho+\lambda-1=7$ and
$$(\mu_a(aababacacaa))_{1,2}=
\{\,(\emptyset,0),(\{a\},0),(B,0),(B,C),(0,C),(0,\{a\}),(0,\emptyset)\,\}\,.$$

\medskip
We can modify this example for arbitrary $\rho,\lambda\geq 4$ as follows:
$$B=\{a,b_1,\dots,b_{\rho-3}\},\ C=\{a,c_1,\dots,c_{\lambda-3}\},\ 
A=B\cup C\,.$$
Then
$$(\mu_a(aab_1ab_2a\dots ab_{\rho-3}ab_1ac_1ac_{\lambda-3}a\dots ac_2ac_1aa)_{1,2}
=$$
$$=
\{\,(\emptyset,0),(\{a\},0),(\{a,b_1\},0),(\{a,b_1,b_2\},0),\dots,(B,0),
(B,C),(0,C),\dots$$
$$\dots,(0,\{a,c_1,c_2\}),(0,\{a,c_1\})(0,\{a\}),
(0,\emptyset)\,\}\,.$$
\end{example}

\section{Level of Varieties}

Let $\mathcal V$ be a variety of languages. A well known fact is that
the pseudovariety of monoids corresponding to the class
$\mathsf{BPol}(\mathbf V)=\bigcup_{k\geq 0}\mathsf{BPol}_k(\mathbf V)$
is generated by all Sch\"utzenberger products
$\Diamond(M_0,\dots,M_n)$
where $M_0,\dots,M_n$ are syntactic monoids of  languages from   $\mathcal V$
-- see~(\cite{pi-kniha}, Theorems 5.1.4. and 5.1.5.).
Of course being interested in $\mathsf{BPol}_k(\mathbf V)$, one takes
$n=k$.

Here we are looking for a single finite monoid recognizing all languages
in $\mathcal V(A)$, $A$ fixed.
We can succeed under certain circumstances as follows.
Let $\mathbf V$ be a locally finite variety of monoids, i.e. the 
finitely generated monoids in $\mathbf V$ are finite. Let $\sim$ be the 
corresponding
fully invariant congruence on $X^*$, $X=\{x_1,x_2,\dots\}$, i.e.
the set of all identities which hold in  $\mathbf V$.
Notice that $X^*/{\sim}$ is the free monoid in $\mathbf V$ over the set
$X$. The
finite members 
of $\mathbf V$ form a
(the so-called equational) pseudovariety of finite monoids.
We denote the corresponding variety of languages by 
$\mathcal V$, i.e. $L\in \mathcal V(A)$
if and only if the syntactic monoid of $L$ is a member of 
$\mathbf V$. Then the free monoid 
in $\mathbf V$ over the set $A$ is the smallest monoid recognizing
all languages from $\mathcal V(A)$. Thus we consider somehow the descriptional
complexity for the whole varieties of languages.

One of the
main results of \cite{kp-cai} was an effective description of the fully
invariant congruence  
$\sim_k$ for the variety
$\mathsf{BPol}_k(\mathbf V)$.
Here we treat only the case of $k=1$.

\begin{result}[(\cite{kp-cai}, Theorem 1)]\label{res-kp}
For $u,v\in A^*$, we have
$$u\sim_1 v \text{ if and only if } u\sim v \text{ and for each } a\in A,$$
$$
\{\,(u'{\sim},u''{\sim})\mid u=u'au'',\, u',u''\in A^*\,\}=
\{\,(v'{\sim},v''{\sim})\mid v=v'av'',\, v',v''\in A^*\,\}.
$$
\end{result}

\begin{proposition} Let $A=\{a_1,\dots,a_d\}\subseteq X$ and
let $$\xi : A^* \rightarrow (A^*/{\sim}\ \Diamond\ A^*/{\sim})\times \dots
\times (A^*/{\sim}\ \Diamond\ A^*/{\sim})\ (d \text{ times })$$
be given by
$$u\mapsto (\mu_{a_1}(u),\dots,\mu_{a_d}(u))\,.$$
Then $\xi(A^*)$ is isomorphic to $A^*/{\sim}_1$, i.e. to the
free monoid in $\mathsf{BPol}_1(\mathbf V)$ over the alphabet $A$. 

In particular, if the cardinality of $A^*/{\sim}$ is $n$, then
the size of $A^*/{\sim}_1$ is bounded by the number $n 2^{dn^2}$.

\end{proposition}

\begin{proof}
The first part follows immediately from Result~\ref{res-kp}.
To get the estimate,
 realize that all the diagonal entries in the matrices
$\mu_{a_1}(u),\dots,\mu_{a_d}(u)$, for a given $u\in A^*$,
are the same.
\end{proof}

Let us consider the simplest non-trivial example.
It shows, among others, that the estimate from the last proposition can be  
far from being optimal.

\begin{example}
Let $\mathbf V = \mathbf {SL}$ -- the class of all semilattices.
Then $u \sim v$ if and only if $\mathsf{c}(u)=\mathsf{c}(v)$.
The free semilattice (in the signature of monoids)
over a set $A\subseteq X$ is isomorphic to
$M=(2^A,\cup)$. In particular, this variety is locally finite.
For the corresponding variety of languages $\mathcal V$
and a finite alphabet $A$, the set $\mathcal V(A)$ consists of unions
of $\overline B$'s, $B\subseteq A$.

Let $A=\{a,b\}$.
We are going to improve
the bound $4\cdot 2^{2\cdot4^2}$ from the last proposition. 
Clearly, the cardinality of $M\Diamond M$ is $2^{20}$.
We will calculate the image of $\mu_a$ first.

\def\c{\mathsf c}
\def\0{\emptyset}
\def\a{\{a\}} \def\b{\{b\}}
\def\ch{\mathsf{char}}
\def\h{\mathsf h} \def\t{\mathsf t}
We write also, for $a_1,\dots,a_k\in A$,
$\mathsf{h}(a_1\dots a_k)=a_1$ and
$\mathsf{t}(a_1\dots a_k)=a_k$.
Let $u=u_0au_1a\dots au_k$ where $u_0,\dots,u_k\in \{b\}^*$. 
The {\it characteristic sequence} $\mathsf{char}(u)$ of $u$ is
$$(\, (\c(u_0),\c(u_1a\dots au_k)),(\c(u_0au_1),\c(u_2a\dots au_k)),
\dots,(\c(u_0a\dots au_{k-1}),\c(u_k))\,)$$
with removed repetitions. We get $(\mu_a(u))_{1,2}$ when considering
it as a set. Note that
$(\mu_a(u))_{1,1}= (\mu_a(u))_{2,2}=\c(u)$ for each $u\in A^*$.
We divide the elements of $A^*$ into several classes:

(i) For $u=1$ we have $\c(u)=\emptyset,\,\ch(u)=1$ (the sequence of length 0).

(ii) For $u=b^k,\,k\geq 1$, we have $\c(u)=\b,\, \ch(u)=1$.

(iii) For $u=a^k,\,k\geq 1$, we have $\c(u)=\a$
and $\ch(a)=((\0,\0))$, $\ch(a^2)=((\0,\a),(\a,\0))$,
and $\ch(a^k)=((\0,\a),(\a,\a),(\a,\0))$ if $k\geq 3$.

All remaining words have $\c(u)=A$.

(iv) If $|u|_a=1$, then
$\ch(u)$ is one of the following sequences
$((\0,\b)),\, ((\b,\0)),\,  ((\b,\b))$.
 
All remaining words have $|u|_a\geq 2$.

(v) If $\h(u)=a,\,\t(u)=b,\, ba$ not being a subword of $u$, i.e.
$u=a^kb^{\ell},\, k\ge 2,\, \ell\ge 1$. Then
either $\ch(u)=(\,(\0,A),(\a,\b)\,)$ for $k=2$ or
$\ch(u)=(\,(\0,A),(\a,A),(\a,\b)\,)$ for $k\ge 3$.

(vi) The case  $u=b^{\ell}a^k,\, k\ge 2,\, \ell\ge 1$ is left-right dual to (v).

(vii) If $\h(u)=a,\,\t(u)=b$,\ $ba$ being a subword of $u$, then
$\ch(u)$ is a subsequence of
$$(\,(\0,A),(\a,A),(A,A),(A,\b)\,)$$
containing the first and the last item. The following words
witness that all 4 possibilities can happen:
$abab, ababab, aabab, aababab$.

(viii)  The case left-right dual to (vii).

(ix) If $\h(u)=\t(u)=a$, then $\ch(u)$ is a subsequence
of
$$(\,(\0,A),(\a,A),(A,A),(A,\a),(A,\0)\,)$$
containing the first and the last item. The following words
witness that all 8 possibilities can happen:
$aba, aaba, abaa, ababa, aabaa,aababa, ababaa, aababaa$.

(x) If $\h(u)=\t(u)=b$, then $\ch(u)=(\,(\b,A),(A,\b)\,)$
or   $\ch(u) =(\,(\b,A),(A,A),(A,\b)\,)$.
Appropriate words are $baab$ and $baaab$.
\smallskip

Altogether we have 30 elements in $\mu_a(A^*)$. 
In fact our consideration until now could be presented in Section 3.
Returning to the free monoid in the variety corresponding to the class
$\mathsf{BPol}_1(\mathbf {SL})$ over $A$, we can state at present only that it
has at most $30\cdot 30$ elements. 
When considering the mapping
$\xi$, not all possible $900$ combinations can happen and we can
further decrease the estimate for $|\xi(A^*)|$. Using more advanced techniques
we can even get 100 as an upper bound.
\end{example}
\medskip

\noindent
{\bf Acknowledgement.} The authors would like to express their gratitude
to Janusz Brzozowski who suggested them to use the construction from \cite{yzs}
in the proof of Proposition~2.


\begin{thebibliography}{00}\label{biography}

\bibitem{b}
J. Brzozowski, Quotient complexity of regular languages,
in {\it Proc. 11th International Workshop
on Descriptional Complexity of Formal Systems (DCFS 2009)},
arXiv:0907.4547v1
       
\bibitem{kp-cai}
O. Klíma and  L. Polák,
Polynomial operators on classes of regular languages,
in {\it Proc.
International Conference on Algebraic Informatics 2009,
Thessaloniki},
Springer LNCS 5725, pp. 260--277

\bibitem{pi-kniha}
J.-E. Pin,
{\it Varieties of Formal Languages},
North Oxford Academic, Plenum, 1986

\bibitem{pi-kap}
J.-E. Pin,
Syntactic semigroups, Chapter 10 in {\it Handbook of Formal Languages},
G. Rozenberg and A. Salomaa eds, Springer, 1997

\bibitem{y}
S. Yu,  State complexity of regular languages. 
{\it J. Autom., Lang. and Comb.} {\bf 6} (2001), pp. 221--234.

  
\bibitem{yzs}
S. Yu, Q. Zhuang and K. Salomaa, The state complexities of some basic
operations on regular languages, 
{\it Theoretical Computer Science} {\bf 125} (1994), pp.
315­--328

\end{thebibliography}
\end{document}